\documentclass[a4paper,12pt]{article} 
\usepackage{amsmath}
\usepackage{amssymb,amsfonts,amsthm}
\usepackage{gastex}
\usepackage[english]{babel}
\usepackage{graphicx}
\usepackage{cite}
\usepackage[margin=1in]{geometry}            
\theoremstyle{plain}
\newtheorem{thrm}{Theorem}[section]
\newtheorem{lmm}{Lemma}[section]
\newtheorem{rmrk}{Remark}[section]
\newtheorem{bsrvtn}{Observation}[section]

\def\IR{{\sf Gr}}
\def\e{{\sf e}}
\def\u{\mathbf{t}}
\def\v{\bar{\mathbf{t}}}

\begin{document}
\title{Binary Patterns in Binary Cube-Free Words: Avoidability and Growth}
%
\author{Robert Mercas\thanks{Otto-von-Guericke-Universit{\"a}t Magdeburg, Fakult{\"a}t f{\"u}r Informatik, PSF 4120, D-39016 Magdeburg, Germany; Supported by the Alexander von Humboldt Foundation; robertmercas@gmail.com}, Pascal Ochem\thanks{CNRS, LIRMM, France; Pascal.Ochem@lirmm.fr}, 
Alexey V. Samsonov\thanks{Ural Federal University, Ekaterinburg, Russia; vonosmas@gmail.com},\\ and Arseny M. Shur\thanks{Ural Federal University, Ekaterinburg, Russia; Arseny.Shur@usu.ru}}
\date{}
\maketitle

\begin{abstract}
The avoidability of binary patterns by binary cube-free words is investigated and the exact bound between unavoidable and avoidable patterns is found. All avoidable patterns are shown to be D0L-avoidable. For avoidable patterns, the growth rates of the avoiding languages are studied. All such languages, except for the overlap-free language, are proved to have exponential growth. The exact growth rates of languages avoiding minimal avoidable patterns are approximated through computer-assisted upper bounds. Finally, a new example of a pattern-avoiding language of polynomial growth is given.
\end{abstract}
%
%

\section{Introduction}
Factorial languages, i.e., languages closed under taking factors of their words, constitute a wide and important class. Each factorial language can be defined by a set of forbidden (avoided) structures: factors, patterns, powers, Abelian powers, etc. In this paper, we consider languages avoiding sets of patterns. 

Pattern avoidance is one of the classical topics in combinatorics of words. Recall that patterns are words over the auxiliary alphabet of variables. These variables admit arbitrary non-empty words over the main alphabet as values. A word over the main alphabet \emph{meets} the pattern if some factor of this word can be obtained from the pattern by assigning values to the variables, and \emph{avoids} the pattern otherwise. 

The main question concerning the avoidance of any set of forbidden structures is whether the language of all avoiding words over the main alphabet is finite or infinite. The set is called \emph{unavoidable} in the first case and \emph{avoidable} in the second case. We use the terms $k$-\emph{(un)avoidable} to specify the cardinality of the main alphabet. 

If a set of structures is avoidable, then the second question is how big is the avoiding language in terms of growth. In general, a simple constraint usually defines either a finite language or a language of exponential growth. So, the examples of languages having subexponential (e.\,g., polynomial) growth are quite valuable. 

For languages avoiding patterns, the main question is far from being satisfactorily answered even for the case of a single pattern. A complete description of the pairs (alphabet, pattern) such that the pattern is avoidable over the alphabet is known only for patterns with at most three variables \cite{Thue12, GoVa91, Roth92, Cas93a, Cas94dis, Och06} and for the patterns that are not avoidable over any alphabet \cite{BEM79, Zim82}. There are very few papers about avoidable sets of patterns; we only mention a result by Petrov \cite{Pet88}. The only exception is the set $\{\sf xxx,xyxyx\}$, defining the binary \emph{overlap-free} language which is quite well presented in literature starting from the seminal paper by Thue \cite{Thue12}.

There are some scattered results concerning the second question (cf. \cite{BeGo07,Och10}). To the best of our knowledge, the only example of a pair (alphabet, pattern) such that the language avoiding the pattern over the alphabet grows subexponentially with the length, was found in \cite{BMT89}: a 7-ary pattern avoidable over the quaternary alphabet. All infinite languages avoiding a binary pattern grow exponentially (combined \cite{Bra83, GoVa91}). However, the binary overlap-free language has polynomial growth \cite{ReSa85}. 

\smallskip
In this paper we start a systematic study of both questions formulated above for the languages specified by a pair of forbidden patterns. It is quite natural to begin with the binary main alphabet and consider the patterns of two variables also. For the first step, it is also natural to fix one of the patterns to be $\sf xxx$, which is the shortest pattern avoidable over two letters. This step is in line with other studies of binary cube-free words with additional constraints (see, e.g., \cite{BCC11}). In this setting, the aim of this paper is to describe the avoidability of binary patterns by the binary cube-free words and the order of growth of avoiding languages. This description is given by the following theorem. Recall that an avoidable set of structures is called \emph{D0L-avoidable} if it is avoided by an infinite word generated by the iteration of a morphism.

\begin{thrm}[Main theorem] \label{main}
Let $P\in\{{\sf x,y}\}^*$ be a binary pattern.\\
1) The set $\{{\sf xxx},P\}$ of patterns is 2-avoidable if and only if $P$ contains as a factor at least one of the words 
\begin{equation} \label{mainlist}
\sf xyxyx,xxyxxy,xxyxyy,xxyyxx,xxyyxyx,xyxxyxy,xyxxyyxy,
\end{equation}
considered up to negation and reversal.\\
2) All 2-avoidable sets $\{{\sf xxx},P\}$ are 2-D0L-avoidable.\\
3) For all 2-avoidable sets $\{{\sf xxx},P\}$, except for the set $\{{\sf xxx,xyxyx}\}$, the avoiding binary language has exponential growth.
\end{thrm}

This is an ``aggregate'' theorem, the proof of which does not follow a single main line but uses quite different techniques. So, we present this proof as a sequence of lesser theorems. Some of these theorems contain refinements to the main theorem (e.\,g., lower bounds for the growth rates of avoiding languages).

Statement 3 of Theorem~\ref{main} leaves little hope to find a subexponentially growing binary language avoiding a pair of patterns; so, we finish the paper by showing an example of such a language avoiding a triplet of binary patterns.

\smallskip
The text is organized as follows. After necessary preliminaries, in Sect.~\ref{sect:avoid} we prove statement~1 of Theorem~\ref{main}; our proof immediately implies statement~2. In Sect.~\ref{sect:low} we finish the proof of Theorem~\ref{main}, exhibiting exponential lower bounds for the cube-free languages avoiding the pattern $\sf xyxyxx$ and all patterns from \eqref{mainlist}, except for the pattern $\sf xyxyx$. In Sect.~5 we estimate actual growth rates of avoiding languages through the upper bounds obtained by computer. Finally, in Sect. 6 we give a new example of a language of polynomial growth. This language consists of cube-free words avoiding a pair of binary patterns.

\section{Preliminaries}

We study finite, right infinite, and two-sided infinite sequences over the main alphabet $\{0,1\}$ and call them \emph{words}, \emph{$\omega$-words}, and \emph{Z-words}, respectively. We also consider \emph{patterns}, which are words over the alphabet of variables $\{\sf{x,y}\}$. Standard notions of \emph{factor}, \emph{prefix}, and \emph{suffix} of a word are used. For a word $w$, we write $|w|$ for its length, $w[i]$ for its $i$th letter, and $w[i...j]$ for its factor starting in the $i$th position and ending in the $j$th position. Thus, $w=w[1...|w|]$. Letters in an $\omega$-word are numbered starting with 1. For a binary word or pattern $w$, its \emph{negation} is the word (resp., pattern) $\bar w$ such that $|w|=|\bar w|$ and $w[i]\ne\bar w[i]$ for any $i$. The \emph{reversal} of $w$ is the word $w[|w|]\cdots w[1]$. A word $w$ has \emph{period} $p$ if $w[1...|w|{-}p]=w[p{+}1...|w|]$. The \emph{exponent} of a word is the ratio between its length and its minimal period. A word is \emph{$\beta$-free} if the exponent of any of its factors is less than $\beta$. Two words are \emph{conjugates} if they can be represented as $uv$ and $vu$, for some words $u$ and $v$. If a word $uv$ has an integer exponent greater than 1, then $vu$ has the same exponent.

A \emph{language} is just a set of words. A language is \emph{factorial} if it is closed under taking factors of its elements. Any factorial language $L$ is determined by its set of \emph{minimal forbidden words}, i.\,e., the words that are not in $L$ while all their proper factors are in $L$. The \emph{growth rate} of a factorial language $L$ is defined as $\IR(L)=\lim_{n\to\infty} (C_L(n))^{1/n}$, where $C_L(n)$ is the number of words of length $n$ in $L$. An infinite language $L$ grows exponentially [subexponentially] if $\IR(L)>1$ [resp., $\IR(L)=1$]. A word $w$ is said to be \emph{(two-sided) extendable} in the language $L$ if $L$ contains, for any $n$, a word of the form $uwv$ such that $|u|,|v|\ge n$. The set of all extendable words in $L$ is denoted by $\e(L)$. 

A \emph{morphism} is any map $f$ from words to words satisfying the condition $f(w)=f(w[1])\cdots f(w[|w|])$ for each word $w$. A morphism is \emph{non-erasing} if the image of any non-empty word is non-empty, and \emph{$n$-uniform} if the images of all letters have length $n$. An $n$-uniform morphism $f$ is called \emph{$k$-synchronizing} if for any factor of length $k$ of any word $f(w)$, the starting positions of all occurrences of this factor in $f(w)$ are equal modulo $n$.

A word $w$ \emph{meets} a pattern $P$ if an image of $P$ under some non-erasing morphism is a factor of $w$; otherwise, $w$ \emph{avoids} $P$. The images of the pattern $\sf xx$ [resp., $\sf xxx$; $\sf xyxyx$] are called \emph{squares} [resp., \emph{cubes}, \emph{overlaps}]. The words avoiding $\sf xx$ [resp., $\sf xxx$; both $\sf xxx$ and $\sf xyxyx$] are \emph{square-free} [resp., \emph{cube-free}, \emph{overlap-free}]. 

If $f$ is a non-erasing morphism and $f(a)=au$ for a letter $a$ and a non-empty word $u$, then an infinite iteration of $f$ generates an $\omega$-word denoted by $\mathbf{f}=f^{\infty}(a)$. The $\omega$-words obtained in this way are called \emph{D0L-words} or \emph{purely morphic} words. The images of letters under a morphism $f$ are called $f$-\emph{blocks}. Note that the D0L-word $\mathbf f$ is a product of $f$-blocks, and also of $f^n$-blocks for any $n>1$, because the morphism $f^n$ generates the same D0L-word $\mathbf f$.

The \emph{Thue-Morse morphism} is defined by the rules $\theta(0)=01$, $\theta(1)=10$ and generates the \emph{Thue-Morse word} $\u=\theta^{\infty}(0)$. The factors of $\u$ are \emph{Thue-Morse factors}. We use the notation $\u_k=\theta^k(0)$ and $\v_k=\theta^k(1)$ for \emph{$\theta^k$-blocks}. The properties listed in Lemma~\ref{TM} below are well known and follow by induction from the facts that $\u$ is a product of $\theta$-blocks and $\theta(\u)=\u$. The third property was first proved by Thue \cite{Thue12}. In the same paper, Thue proved that $\u$ is an overlap-free word.

\begin{lmm} \label{TM}
1) The number of Thue-Morse factors of length $n$ is $\Theta(n)$.\\
2) For any fixed $k$, the number of pairs of equal adjacent $\theta^k$-blocks in any Thue-Morse factor of length $n$ is $n/(3{\cdot}2^k)+O(1)$.\\
3) If $vv$ is a Thue-Morse factor, then $v$ is either a $\theta^k$-block or a product of three alternating $\theta^k$-blocks, for some $k\ge0$. The position in which $vv$ ends in $\u$ is divisible by $2^k$ but not by $2^{k+1}$.
\end{lmm}

A set of patterns (in particular, a single pattern) is \emph{2-avoidable} if there exists a binary $\omega$-word avoiding this set, and \emph{2-D0L-avoidable} if such a D0L-word exists. The existence of an avoiding $\omega$-word is clearly equivalent to the existence of an infinite set of avoiding finite words.

\section{Avoidable and unavoidable patterns} \label{sect:avoid}

In this section we classify the binary patterns avoidable by binary cube-free words. As was already mentioned, the pattern $\sf xyxyx$ is avoided by the Thue-Morse word. The following observation can be easily checked by hand or by computer.

\begin{bsrvtn}
All binary patterns of length at most 5, except for the pattern $\sf xyxyx$, are unavoidable by binary cube-free words.
\end{bsrvtn}

Next we focus our attention on the patterns of length 6. For both avoidability and growth, the patterns can be studied up to negation and reversal. Thus, we obtain the list of eight patterns:
\begin{equation} \label{list6}
\sf xxyxxy,xxyxyx,xxyxyy,xxyyxx,xxyyxy,xyxxyx,xyxyyx,xyyxxy.
\end{equation}

The pattern $\sf xxyxyx$ is obviously avoided by the Thue-Morse word as it has the factor $\sf xyxyx$. The pattern $\sf xxyxxy$ is also avoided by the Thue-Morse word, as was first mentioned in \cite{Cas93a}. (For the complete set of binary patterns avoided by the Thue-Morse word see \cite{Sh96sem}.)
The last four words from the list \eqref{list6} are unavoidable, as can be easily checked by computer. The longest cube-free words avoiding these patterns are listed in Table~\ref{lengths}. The remaining two patterns $\sf xxyyxx$ and $\sf xxyxyy$ are avoidable, see Theorems~\ref{avoid1} and~\ref{avoid2} below.

It follows immediately from the classification of patterns of length 6 that almost all binary patterns of length 7 are avoidable. Only three patterns of length 7, namely, 
$$
\sf xxyyxyx,xyxxyxy,xyxxyyx,
$$
have no proper avoidable factors. The last of these patterns is unavoidable (see Table~\ref{lengths}), while the first two are avoidable (see Theorem~\ref{avoid2}). Finally, there is a unique pattern $\sf xyxxyyxy$ of length 8 for which all proper prefixes and suffixes are unavoidable. But this pattern is avoidable by the Thue-Morse word \cite{Sh96sem}. 

\begin{table}[htb]
\caption{Longest avoiding cube-free words for unavoidable patterns.} \label{lengths}
\centerline{
\begin{tabular}{|l|l|l|}
\hline
Pattern&Longest avoiding cube-free word $u$&$|u|$\\
\hline
{\sf xxyyxy}&\small 010100101101001011010010011001100&33\\
{\sf xyxxyx}&\small 00110101100101001101011001001101100101001101011001010011&56\\
{\sf xyxyyx}&\small 001100100110110010011011001001011&33\\
{\sf xyyxxy}&\small 0011011010010100101101100&25\\
{\sf xyxxyyx}&\small 0011001100100101101001011010010100101101100&43\\
\hline
\end{tabular}}
\end{table}

Thus, we have reduced statement 1 of Theorem~\ref{main} to the proof of Theorems~\ref{avoid1} and~\ref{avoid2}. Since all avoidability proofs are obtained by constructing D0L-words, we also get statement 2 of Theorem~\ref{main}.

\begin{thrm} \label{avoid1}
There exists a binary cube-free D0L-word avoiding the pattern $\sf xxyyxx$.
\end{thrm}

Consider the morphism $\mu$ defined by the equalities $\mu(0)=010$, $\mu(1)=011$\footnote{The morphism $\mu'$ defined by $\mu'(0)=001$, $\mu'(1)=011$, 
also avoids $\{\sf xxx,xxyyxx\}$ (see \cite{SaSh12b}; independently discovered by J. Shallit, private communication). We prefer the morphism $\mu$ because its study allows us to prove that the avoiding language grows exponentially (see Theorem~\ref{xxyyxx}).}, and the D0L-word $\mathbf{m}=\mu^{\infty}(0)$. Some properties of the word $\mathbf{m}$ are gathered in the following lemma.

\begin{lmm} \label{muprop}
Let $k$ be an arbitrary nonnegative integer.\\
1) One has $\mathbf{m}[3k{+}1]=0$ and $\mathbf{m}[3k{+}2]=1$.\\
2) The last letter in the block $\mu^k(a)$ is $a$. All other letters in $\mu^k(0)$ and $\mu^k(1)$ coincide.\\
3) If $u$ is a factor of $\mathbf{m}$ and $|u|=3^k$, then the starting positions of all occurrences of $u$ in $\mathbf{m}$ are equal modulo $3^k$.\\
4) If $\mathbf{m}$ contains a square $uu$ and $3^k\le|u|<3^{k+1}$, then $|u|\in\{3^k,2\cdot3^k\}$.\\
5) Suppose that $\mathbf{m}[r_13^k{+}c\ldots r_23^k{+}c{-}1]$ is a square for some integers $r_1,r_2,c$ such that $0<c\le3^k$. Then the word $\mathbf{m}[r_13^k{+}1\ldots r_23^k]$ is a square as well.
\end{lmm}

\begin{proof}
Properties 1 and 2 follow immediately from the definition of $\mu$. Let us prove property 3 by induction on $k$.

The base cases are $k=0$ (holds trivially) and $k=1$, which follows directly from property~1. Now we let $k\ge2$ and prove the inductive step. Assume to the contrary that two occurrences of the factor $u$ of length $3^k$ have starting positions $j_1$ and $j_2$ that are different modulo $3^k$. These positions are also the starting positions of the occurrences of the factor $u[1\ldots3^{k-1}]$. Hence, $j_1\equiv j_2\pmod{3^{k-1}}$ by the inductive assumption. By property~2, both considered occurrences of $u$ are preceded by the same $(j_1\bmod3^{k-1})-1$ letters. Thus, $\mathbf{m}$ contains a factor $u'$ such that $|u'|=3^k$, $u'$ is a product of $\mu^{k-1}$-blocks, and the starting positions of two occurrences of $u'$ are different modulo $3^k$. Hence, $\mathbf{m}$ also contains the factor $\mu^{-1}(u')$ of length $3^{k-1}$ such that the starting positions of two occurrences of $\mu^{-1}(u')$ are different modulo $3^{k-1}$, in contradiction with the inductive assumption. Therefore, property~3 is proved.

Property 4 is an immediate consequence of property 3. In order to prove property~5, we note that property~2 implies the equality $\mathbf{m}[r_13^k{+}1\ldots r_13^k{+}c{-}1]=\mathbf{m}[r_23^k{+}1\ldots r_23^k{+}c{-}1]$. Hence, the two considered words are conjugates. But all conjugates of a square are squares.
\end{proof}

\begin{proof}[Proof of Theorem~\ref{avoid1}]
Let us prove that $\mathbf{m}$ is cube-free and avoids $\sf xxyyxx$. Aiming at a contradiction, first assume that $\mathbf{m}$ contains a cube; consider the shortest one, say $u^3$. Then $|u|>2$ in view of Lemma~\ref{muprop},\,1. Hence $|u|\equiv0\pmod3$ Lemma~\ref{muprop},\,4. Using Lemma~\ref{muprop},\,5, we find a cube $u'^3$ which is a product of $\mu$-blocks. Then $\mathbf{m}$ contains the cube $(\mu^{-1}(u'))^3$, in contradiction with the choice of $u^3$.

The argument for the pattern $\sf xxyyxx$ is essentially the same. If $\mathbf{m}$ has a factor $uuvvuu$, then Lemma~\ref{muprop},\,1 implies that at least one of the numbers $|u|,|v|$ is greater than 2. Then this number is divisible by 3 by Lemma~\ref{muprop},\,4, and hence the other number is divisible by 3 too (Lemma~\ref{muprop},\,3). Therefore, we can apply Lemma~\ref{muprop},\,5 to get a factor $u'u'v'v'u'u'$ which begins with the starting position of a $\mu$-block. Then $\mathbf{m}$ contains a shorter forbidden factor $\mu^{-1}(u'u'v'v'u'u')$, contradicting to the choice of $uuvvuu$.
\end{proof}

\begin{thrm} \label{avoid2}
There exist binary cube-free D0L-words avoiding the patterns $\sf xxyxyy$, $\sf xxyyxyx$, and $\sf xyxxyxy$, respectively\footnote{2-D0L-avoidability of the pattern $\sf xxyxyy$ was first observed by J. Cassaigne who found a 12-uniform avoiding cube-free morphism (private communication). This pattern is also avoided by a cube-free ``quasi-morphism'' defined in \cite{SaSh12b}.}. 
\end{thrm}

Our proof involves a rather short computer check based on the following two lemmas.

\begin{lmm}[Richomme, Wlazinski,\cite{RiWl00}] \label{3free}
A morphism $f:\{0,1\}\to\{0,1\}$ is cube-free if and only if the word $f(001101011011001001010011)$ is cube-free.
\end{lmm} 

\begin{lmm} \label{synch}
Suppose that an $\omega$-word $\mathbf f$ is generated by a $k$-synchronizing $n$-uniform cube-free binary morphism $f$, and $P\in\{\sf xxyxyy, xyyxyx\}$. Then $\mathbf f$ meets $P$ if and only if $\mathbf f$ contains the factor $g(P)$ for some morphism $g$ satisfying $|g({\sf x})|,|g({\sf y})|<k$.
\end{lmm} 

\begin{proof}
We assume that the word $\mathbf f$ contains a factor of the form $g(P)$ such that $\max\{|g({\sf x})|,|g({\sf y})|\}\ge k$ and prove that $\mathbf f$ must contain a shorter image of $P$. Let $x'=g({\sf x}),y'=g({\sf y}), |x'|\ge k$. The starting positions of all occurrences of $x'$ in $\mathbf f$ are equal modulo $n$ by the definition of $k$-synchronizing morphism. Considering the occurrences inside $g(P)$, we see that $|x'|\equiv|x'y'|\equiv0\pmod n$ if $P=\sf xxyxyy$ and $|x'y'y'|\equiv|x'y'|\equiv0\pmod n$ if $P=\sf xyyxyx$. Thus, in both cases $|x'|$ and $|y'|$ are divisible by $n$. The assumption $|y'|\ge k$ leads to the same result.

Now we can write $x'=x_1x_2x_3, y'=y_1y_2y_3$, where $x_1,y_1$ [respectively, $x_2,y_2$; $x_3,y_3$] are suffixes [respectively, products; prefixes] of $f$-blocks, $|x_1|=|y_1|=r$, $|x_3|=|y_3|=l$, $l+r=n$. An $f$-block is determined either by its prefix of length $l$ or by its suffix of length $r$. Thus, $\mathbf f$ contains another image of $P$ of length $|g(P)|$: the starting position of this image is either $r$ symbols to the right or $l$ symbols to the left from the starting position of $g(P)$. This new image $h(P)$ is a product of $f$-blocks. As a result, $h({\sf x})$ and $h({\sf y})$ are products of $f$-blocks also. Hence, $\mathbf f$ contains an image of $P$ under the composition of $f^{-1}$ and $h$; this image is shorter than $g(P)$, as required.
\end{proof}

\begin{proof}[Proof of Theorem~\ref{avoid2}]
Consider the morphisms $h_1$, $h_2$, and $h_3$ such that
\begin{equation*}
\arraycolsep=3mm
\begin{array}{lll}
h_1(0)= 0110010&h_2(0)= 01001&h_3(0)= 010011\\
h_1(1)= 1001101&h_2(1)= 10110&h_3(1)= 011001
\end{array}
\end{equation*}
Checking the condition of Lemma~\ref{3free} by computer, we obtain that all these morphisms are cube-free. Furthermore, it can be directly verified that $h_1$, $h_2$, and $h_3$ are 6-, 6-, and 5-synchronizing, respectively. Hence, if the D0L-word $\mathbf h_1$ generated by $h_1$ meets the pattern $P=\sf xxyxyy$, then by Lemma~\ref{synch}, $\mathbf h_1$ contains an image of $P$ of length at most $5\cdot6=30$. Thus, it is enough to check all factors of $\mathbf h_1$ of length at most 30. Any such factor is contained in the image of a factor of $\mathbf h_1$ of length 6; this factor, in turn, belongs to the image of a factor of length 2, while all factors of length 2 can be found in $h_1(0)$. Therefore, we just need to examine all factors of length up to 30 in the word $h_1^3(0)$. A computer check shows that there are no images of $P$ among such factors. So, we conclude that $\mathbf h_1$ avoids both cubes and the pattern $\sf xxyxyy$. 

Similar argument for the morphism $h_2$ and the pattern $\sf xxyyxyx$, containing $\sf xyyxyx$, shows that it is enough to examine the factors of length up to 35 in the word $h_2^4(0)$. A computer check implies the desired avoidability result. In the same way, we check the factors of length up to 28 in $h_3^3(0)$ to show that the corresponding D0L-word avoids the pattern $\sf xyxxyxy$. The theorem is proved.
\end{proof}

\section{Lower bounds for the growth rates} \label{sect:low}

In this section we prove lower bounds for the growth rates of the languages avoiding the sets $\{{\sf xxx},P\}$, where $P=\sf xyxyxx$ or $P$ is any of the patterns listed in \eqref{mainlist}, except for the pattern $\sf xyxyx$. In particular, the results of this section imply statement~3 of Theorem~\ref{main}. 

The bounds are obtained using two different methods. The first method uses block replacing in the factors of D0L-words, and is purely analytic. We apply this method to the patterns $\sf xyxyxx, xxyxxy, xxyyxx$. The second method uses morphisms that act on the ternary alphabet and map ternary square-free words to binary cube-free words avoiding the given patterns. This method requires some computer search and check; we apply it to the remaining four patterns. (The second method can be applied for all patterns, but the analytic bounds are a bit better.)

\subsection{Replacing blocks in D0L-words}

\begin{thrm} \label{xyxyxx}
The number of binary cube-free words avoiding the pattern $\sf xyxyxx$ grows exponentially with the rate of at least $2^{1/24}\approx1.0293$.
\end{thrm}

\begin{proof}
Let $L$ be the language of all binary cube-free words avoiding $\sf xyxyxx$. Recall that $L$ contains all Thue-Morse factors. Consider the ``distorted'' $\theta^5$-block
\begin{equation} \label{insert}
\u'=0110\,1001\,1001\,\pmb{1}\,0110\,1001\,0110\,0110\,1001,
\end{equation}
obtained from the block $\u_5$ by inserting the letter 1 in the 13th position, and its negation $\v'$ obtained by inserting a $0$ in the same way into $\v_5$. Let $S$ be the set of all $\omega$-words that can be obtained from the Thue-Morse word $\u$ by replacing some of its $\theta^5$-blocks by the corresponding distorted blocks. Available places for inserting letters are shown below:
\begin{equation} \label{insert2}
\unitlength=0.8mm
\begin{picture}(126,17)(7,0)
\multiput(10.4,2)(90,0){2}{\makebox(0,0)[lb]{\small$\u_2\v_2\v_2\u_2\v_2\u_2\u_2\v_2$}}
\multiput(40.4,2)(30,0){2}{\makebox(0,0)[lb]{\small$\v_2\u_2\u_2\v_2\u_2\v_2\v_2\u_2$}}
\put(7,2.5){\makebox(0,0)[rb]{$\u=$}}
\put(132,2.5){\makebox(0,0)[lb]{\small$\ldots$}}
\put(10,0){\line(1,0){124}}
\put(55,16){\line(1,0){30}}
\multiput(10,0)(30,0){5}{\line(0,1){3}}
\multiput(55,16)(30,0){2}{\line(0,-1){3}}
\multiput(20.8,6.5)(30,0){4}{\makebox(0,0)[cb]{$\downarrow$}}
\put(65.8,6.5){\makebox(0,0)[cb]{$\downarrow$}}
\put(20.8,11.5){\makebox(0,0)[cb]{\small$1$}}
\put(65.8,11.5){\makebox(0,0)[cb]{\small$1$}}
\put(110.8,11.5){\makebox(0,0)[cb]{\small$1$}}
\put(50.8,11.5){\makebox(0,0)[cb]{\small$0$}}
\put(80.8,11.5){\makebox(0,0)[cb]{\small$0$}}
\multiput(25,-0.7)(90,0){2}{\makebox(0,0)[ct]{\small$\u_5$}}
\multiput(55,-0.5)(30,0){2}{\makebox(0,0)[ct]{\small$\v_5$}}
\put(70,16.5){\makebox(0,0)[cb]{\small$\u_5$}}
\end{picture}
\end{equation}
Let $\mathbf{z}\in S$. It is easy to check manually that $\mathbf{z}$ does not contain short cubes; as it will be shown below, $\mathbf z$ does not contain long overlaps, and hence has no cubes at all. Now, our goal is to prove that $\mathbf{z}$ avoids the pattern $\sf xyxyxx$.\\[4pt]
\emph{Claim}. Let $w=uvuvuu$ be a minimal forbidden word for $L$. Then $u\in\{0,1,01,10\}$.\\[4pt]
Assume that $|u|>1$. Then $u[i]\ne u[i{+}1]$ for all $i$ and, moreover, $u[|u|]\ne u[1]$. Indeed, otherwise $w[i...2|uv|{+}i{+}1]$ is an image of $\sf xyxyxx$, a contradiction with the minimality of $w$. Hence, $u\in\{(01)^s,(10)^s\}$. Since $uu$ is not forbidden, $s=1$. The claim is proved.

\smallskip
Let us consider the overlaps in $\mathbf{z}$. The case analysis below is performed up to negation. Each overlap surely contains at least one inserted letter. Two short overlaps can be easily observed inside the word $\u'$, see \eqref{insert}. They are $\u'[5...14]=1001100110$ and $\u'[11...18]=01101101$. These overlaps obviously avoid the pattern $\sf xyxyxx$. One can easily check that there is no other overlap of period $\le10$. Note that the words $0011001\pmb{1}$ and $1\pmb{1}011$ are not Thue-Morse factors and thus their occurrences in $\mathbf{z}$ indicate an inserted letter (the bold one).

Now assume to the contrary that some word $\mathbf{z}\in S$ meets the pattern $\sf xyxyxx$. Let $w=uvuvuu$ be the shortest word among the images of $\sf xyxyxx$ in all words $\mathbf{z}\in S$. We already know that $|uv|>10$. So, if $v$ contains an inserted letter then one of the corresponding ``indicators'' $0011001\pmb{1}$ and $1\pmb{1}011$ occurs inside $uvu$. Hence, the same letter was inserted in the other occurrence of $v$. If we delete both these inserted letters from $w$, we will get a shorter image of $\sf xyxyxx$, contradicting to the choice of $w$. Thus, $w$ contains inserted letters only inside $u$. Recall that $|u|\le2$ by the claim.

Assume that the letter 1 was inserted inside the second (middle) occurrence of $u$. Then if $|u|=1$, the Thue-Morse word contains the square $vv$. If $|u|=2$ then $u=10$, because the inserted letter is preceded by the same letter. So, the 0 in the first (left) occurrence of $u$ is not an inserted letter. Thus, $0v0v$ is a square in $\u$. Then the word $v$ or the word $0v$ should be either a $\theta^k$-block or a product of three alternating $\theta^k$-blocks (Lemma~\ref{TM},\,3). But $v$ ends with $01100110$, see \eqref{insert}, so we get a contradiction. 

Now note that if an inserted 1 is in the third (right) occurrence of $u$, then the corresponding indicator $0011001\pmb{1}$ occurs in the suffix $uvu$ of $w$. Hence $0011001\pmb{1}$ occurs in the prefix $uvu$ of $w$. Thus, 1 was also inserted inside the middle occurrence of $u$, which is impossible as we have shown already. So, the only remaining position for the inserted letter is in the left occurrence of $u$. Then $|u|=1$ (otherwise, $\u$ contains an overlap), and $vuvu$ is a factor of $\u$. But the inserted letter is preceded by the same letter, so $uvuvu$ must be a Thue-Morse factor. This contradiction finishes the proof of the fact that the word $\mathbf{z}$ avoids $\sf xyxyxx$.

\smallskip
Thus, we have proved that all finite factors of the word $\mathbf{z}$ belong to $L$. To finish the proof, we take a large enough number $n$ and consider all Thue-Morse factors of length $n$. For each factor, we perform the insertions of letters into $\theta^5$-blocks according to both \eqref{insert} and the negation of \eqref{insert}, in all possible combinations. Thus we obtain $2^k$ words from $L$, where $k$ stands for the number of $\theta^5$-blocks in the processed factor. Note that the words obtained from different factors are different (for instance, such words contain indicators in different positions). A Thue-Morse factor of length $n$ contains $n/32{+}O(1)$ ``regular'' $\theta^5$-blocks plus those $\theta^5$-blocks occurring on the border of two equal $\theta^5$-blocks, see \eqref{insert2}. Using Lemma~\ref{TM},\,2, we obtain the total of $n/24+O(1)$ blocks. Taking Lemma~\ref{TM},\,1 into account, we see that we constructed $\Theta(n)2^{n/24+O(1)}$ words from $L$, and the lengths of these words cover the interval of length $\Theta(n)$. Therefore, the growth rate of $L$ is at least $2^{1/24}$, as desired.
\end{proof}

\begin{thrm} \label{xxyxxy}
The number of binary cube-free words avoiding the pattern $\sf xxyxxy$ grows exponentially with the rate of at least $2^{1/24}\approx1.0293$.
\end{thrm}

\begin{proof}
As in the proof of Theorem~\ref{xyxyxx}, we get an exponential lower bound using multiple insertions into the Thue-Morse word. But now we need to insert  rather long words, not just letters. Let $L$ be the language of all binary cube-free words avoiding $\sf xxyxxy$. Recall that $L$ contains all Thue-Morse factors. Consider the word
\begin{equation} \label{insert3}
\u'=0110\,1001\,1001\,0110\,1001\,0110\ \pmb{01010\ 01\,1001\,0110\,1001\,0110}\ 0110\,1001,
\end{equation}
obtained from the $\theta^5$-block $\u_5$ by inserting the marked factor $s$ of length 23 in the 25th position, and its negation $\v'$ obtained by inserting $\bar s$ in the 25th position of $\v_5$. One can check directly that both $\u'$ and $\v'$ are cube-free and avoid $\sf xxyxxy$. Note that $\u'[25...29]=01010$, but $\u'[1...28]$ is an overlap-free word ending with the square $\theta(100100)$, and $\u'[26...55]$ is a Thue-Morse factor. Let $S$ be the set of all $\omega$-words that can be obtained from the Thue-Morse word $\u$ by replacing some of its $\theta^5$-blocks by the corresponding blocks $\u'$, $\v'$. Let us consider a successive pair of inserted factors in $\mathbf{z}\in S$ (here $a,b\in\{0,1\}$):
\begin{equation} \label{insert4}
\unitlength=0.8mm
\begin{picture}(105,14)(5,0)
\put(2,1.5){\makebox(0,0)[rb]{\small$\mathbf{z}=$}}
\put(8,1.5){\makebox(0,0)[rb]{\small$\ldots$}}
\put(112,1.5){\makebox(0,0)[lb]{\small$\ldots$}}
\put(10,0){\line(1,0){100}}
\put(27.2,11){\line(1,0){65.6}}
\multiput(27.2,11)(65.6,0){2}{\line(0,-1){3}}
\put(60,11.5){\makebox(0,0)[cb]{\small overlap-free word}}
\put(27.2,6){\line(1,0){61}}
\multiput(27.2,6)(61,0){2}{\line(0,-1){3}}
\put(57.7,6.5){\makebox(0,0)[cb]{\small$w$}}
\put(25,0.5){\makebox(0,0)[lb]{\small$a\,\bar a\,a\,\bar a\,a$}}
\put(95,0.5){\makebox(0,0)[rb]{\small$b\,\bar b\,b\,\bar b\,b$}}
\end{picture}
\end{equation}
We see that $w$ is a Thue-Morse factor, $wb\bar{b}$ is an overlap-free word with the suffix $(\theta(\bar bbb))^2$. Moreover, assume for a moment that the left of the two considered insertions is withdrawn; then \emph{the factor $wb\bar{b}$ still would occur in the same place}.

Let us show that an $\omega$-word $\mathbf z\in S$ contains no overlap except for $01010$ and $10101$. Assume to the contrary that such overlaps exist. Consider the overlap $w=uvuvu$ which has the shortest period (among the overlaps in all $\mathbf z\in S$) and is not extendable (i.\,e., is not contained in a longer factor of $\mathbf z$ with the same period). In view of \eqref{insert4}, $w$ should contain the factor 10101 or 01010. We assume w.l.o.g. that $w$ contains 01010 and $w[i...i{+}4]=01010$ is the rightmost occurrence of this factor in $w$. This occurrence in certainly not inside the prefix $uvu$ of $w$. Suppose that this occurrence is inside the suffix $uvu$ of $w$. Then we have $w[i{-}|uv|...i{-}|uv|{+}4]=01010$. Both these occurrences of 01010 are prefixes of the occurrences of $s$ in $\mathbf z$. Since the leftmost of these occurrences of $s$ is obviously inside $w$, the rightmost one is also inside $w$ due to non-extendability of $w$. Moreover, non-extendability of $w$ implies that the rightmost occurrence of $s$ is not a suffix of $w$, because $s$ is always followed by $\v_3$. Now we can delete both mentioned occurrences of $s$ and get an overlap with a smaller period in contradiction with the choice of $w$. One case of mutual location of the factors of $w$ is depicted below, the others are quite similar. Deleting the occurrences of $s$ in the case presented in the picture gives the overlap $u_2v_1u_2v_1u_2$.
\begin{equation*} 
\unitlength=0.8mm
\begin{picture}(105,12)(5,-1)
\put(5,0.7){\makebox(0,0)[rb]{\small$w=$}}
\put(10,0){\line(1,0){100}}
\multiput(10,0)(45,0){3}{\line(0,1){3}}
\multiput(20,0)(45,0){3}{\line(0,1){3}}
\multiput(37,7)(45,0){2}{\line(1,0){23}}
\multiput(37,7)(45,0){2}{\line(0,-1){3}}
\multiput(60,7)(45,0){2}{\line(0,-1){3}}
\multiput(47.5,7.5)(45,0){2}{\makebox(0,0)[cb]{\small$s$}}
\multiput(12.7,0.5)(45,0){3}{\makebox(0,0)[cb]{\small$u_1$}}
\multiput(17.7,0.5)(45,0){3}{\makebox(0,0)[cb]{\small$u_2$}}
\multiput(29,0.5)(45,0){2}{\makebox(0,0)[cb]{\small$v_1$}}
\multiput(46.5,0.5)(45,0){2}{\makebox(0,0)[cb]{\small$v_2$}}
\end{picture}
\end{equation*}
Thus, it remains to consider the case when the rightmost occurrence of 01010 in $w=uvuvu$ strictly contains the middle $u$. Since $\u'$ contains no overlaps except for $01010$, we conclude that $|s|<|uv|$. Then the mutual location of the factors in $w$ looks like in the following picture.
\begin{equation*} 
\unitlength=0.8mm
\begin{picture}(100,10)(10,0)
\put(5,0.7){\makebox(0,0)[rb]{\small$w=$}}
\put(10,0){\line(1,0){100}}
\put(52,8){\line(1,0){38}}
\multiput(10,0)(47,0){3}{\line(0,1){3}}
\multiput(16,0)(47,0){3}{\line(0,1){3}}
\multiput(52,8)(16,0){2}{\line(0,-1){3}}
\put(90,8){\line(0,-1){3}}
\multiput(13,0.7)(47,0){3}{\makebox(0,0)[cb]{\small$u$}}
\multiput(19.5,0.5)(47,0){2}{\makebox(0,0)[cb]{\small$v_1$}}
\multiput(54.5,0.5)(47,0){2}{\makebox(0,0)[cb]{\small$v_4$}}
\multiput(32,0.7)(47,0){2}{\makebox(0,0)[cb]{\small$v_2$}}
\multiput(47.5,0.7)(47,0){2}{\makebox(0,0)[cb]{\small$v_3$}}
\put(60,7.5){\makebox(0,0)[ct]{\small$0\,1\,0\,1\,0$}}
\put(71,8.5){\makebox(0,0)[cb]{\small$s$}}
\end{picture}
\end{equation*}
The word $v_3$ begins and ends with 0, and $v_4$ also begins with 0, see \eqref{insert3}. Then the word $v_3v_3v_4$ begins with a shorter overlap, and this overlap contains at least five zeroes. Since the word $v_3v_3v_4$ occurs in an $\omega$-word from $S$, we get a contradiction with the minimality of the period of $w$. Thus, we have proved that the ``long'' overlap $w$ does not exist. Therefore, all $\omega$-words from $S$ contain no overlaps except for 01010 and 10101 and, in particular, are cube-free.

\smallskip
Now assume that $\mathbf z\in S$ contains an image of the pattern $\sf xxyxxy$, i.e., the factor $w=uuvuuv$ for some nonempty words $u$, $v$. W.l.o.g., this factor is preceded in $\mathbf z$ by 0. Since the word $0w$ is not an overlap, the word $v$ ends with 1. Then $0w$ contains both factors $0uu$ and $1uu$. But one of the words $0uu$, $1uu$ is an overlap, i.e., is equal to 01010 (resp., 10101). Let $v=v'1$ and consider both cases.

\emph{Case 1.} $0w=0\,0101v'1\,0101v'1$. By \eqref{insert3}, $v'$ ends with 1. Then the word $w$ is followed by 0. Hence, $w0$ is an overlap, which is impossible.

\emph{Case 2.} $0w=0\,1010v'1\,1010v'1$. There is no factor 01010 or 10101 on the border between the left and the right $uuv$. Hence, the factors $s$ and $\bar s$ in $w$, if any, are inside $uuv$ (recall that $w$ is not extendable to the right, because $\mathbf z$ has no long overlaps). Therefore, after deleting all occurrences of $s$ and $\bar s$ in $w$, we will still have a square of the form $(1010...)^2$. But the Thue-Morse word has no such squares, see Lemma~\ref{TM},\,3. This contradiction proves that the $\omega$-word $\mathbf z$ avoids the pattern $\sf xxyxxy$.

\smallskip
It remains to estimate the total number of factors in all words $\mathbf z$. Repeating the argument from the proof of Theorem~\ref{xyxyxx}, we arrive at the same bound $2^{1/24}$.
\end{proof}

\begin{thrm} \label{xxyyxx}
The number of binary cube-free words avoiding the pattern $\sf xxyyxx$ grows exponentially with the rate of at least $2^{1/18}\approx1.0392$.
\end{thrm}

\begin{proof}
Proving this lower bound, we cannot rely on the Thue-Morse word, because it meets the pattern $\sf xxyyxx$. Instead, we apply the insertion technique to the D0L-word $\mathbf{m}$ generated by the morphism $\mu$, introduced in Sect.~\ref{sect:avoid}. The word $\mathbf m$ is a product of $\mu$-blocks, as well as of $\mu^2$-blocks, and has no other occurrences of such blocks by Lemma~\ref{muprop},\,3. Consider the word
\begin{equation}\label{insert01010}
\mathbf{m}'=010\,011\,010\ \pmb{01010}
\end{equation}
obtained by attaching the factor $01010$ to the block $\mu^2(0)$. Let $S$ be the set of all $\omega$-words obtained from $\mathbf{m}$ by replacing some of  the blocks $\mu^2(0)$ by the words $\mathbf{m}'$ (in other words, by inserting the factor 01010 after some blocks $\mu^2(0)$). 
\begin{itemize}
\item[($\vartriangle$)] If one inserts 01010 after $u$ in a word $u\mathbf{v}\in S$, then $u$ is followed by 010 [resp., $\mathbf v$ is preceded by 1010] both before and after insertion.
\end{itemize}
Note that $0101$ is not a factor of $\mathbf{m}$ by Lemma~\ref{muprop},\,1, and hence we use this word as a ``marker''. Let us show that $S$ avoids $\{{\sf xxx, xxyyxx}\}$. Assume to the contrary that some $\mathbf{z}\in S$ has a forbidden factor, and $w$ is the shortest one among all forbidden factors of all words $\mathbf{z}\in S$. Since $w$ is not a factor of $\mathbf{m}$, it contains at least one marker $0101$.


The word $w$ equals either to $uuu$ or to $uuvvuu$, for some words $u,v$. If $u$ or $v$ contains the factor $01010$, then one can cancel the corresponding insertions inside each occurrence of this word, thus getting a shorter forbidden factor in contradiction with the choice of $w$.
Thus, all inserted factors inside $w$ are on the borders of its parts.

Let $w=uuu$. Using the fact that $\mathbf{m}'$ is always followed by a $\mu^2$-block, it is easy to check that $|u|\ge5$. If $uu$ contains $01010$ somewhere in the middle, then $01010=z_1z_2$ and $u=z_2u'z_1$. Hence, after cancelling the two insertions inside $uuu$, one obtains a shorter cube $u'u'u'$, a contradiction. Finally, if $uuu$ ends with $0101$, then this marker is a suffix of $u$, and we get the previous case. Thus, $w$ has no markers, a contradiction.

\smallskip
Now let $w=uuvvuu$. First consider the case where either $u$ or $v$ lies strictly inside some factor $01010$ (and hence, is equal to 01 or 10). If $u=01$, then $vv=0z$, where $z$ is either a product of $\mu^2$-blocks or such a product with $01010$ inserted in the middle. In the first case $0z$ is a factor of $\mathbf{m}$ and hence is not a square by Lemma~\ref{muprop},\,4. In the second case, the right half of $0z$ cannot begin with $00$ as $0z$ itself does; once again we see that $0z$ is not a square. The case $u=10$ and $vv=z0$ is symmetric to the above one.

If $v=01$ [$v=10$], then $u$ begins with 00 [resp., 0] and ends with 0 [resp., 00], implying that $uu$ contains the cube $000$, which is impossible. 

Thus, the factors $01010$ can be found inside $w$ only in the following places:
\begin{equation*} 
\unitlength=0.8mm
\begin{picture}(120,8)(10,0)
\gasset{AHnb=0}
\put(5,0.7){\makebox(0,0)[rb]{\small$w=$}}
\put(10,0){\line(1,0){120}}
\multiput(10,0)(20,0){7}{\line(0,1){3}}
\drawline[dash={0.7 0.7}{0.3}](26,5)(26,8)(34,8)(34,5)
\drawline[dash={0.7 0.7}{0.3}](46,5)(46,8)(54,8)(54,5)
\drawline[dash={0.7 0.7}{0.3}](66,5)(66,8)(74,8)(74,5)
\drawline[dash={0.7 0.7}{0.3}](86,5)(86,8)(94,8)(94,5)
\drawline[dash={0.7 0.7}{0.3}](106,5)(106,8)(114,8)(114,5)
\multiput(20,0.7)(20,0){2}{\makebox(0,0)[cb]{\small$u$}}
\multiput(60,0.7)(20,0){2}{\makebox(0,0)[cb]{\small$v$}}
\multiput(100,0.7)(20,0){2}{\makebox(0,0)[cb]{\small$u$}}
\end{picture}
\end{equation*}
(In addition, $w$ can have the suffix $0101$; in case of any other partial intersection of $01010$ and $w$, the deletion of this occurrence of 01010 from $\mathbf{z}$ leaves $w$ unchanged by ($\vartriangle$).)

Consider any square in $\mathbf{z}$ containing $01010$ in the middle. Such a square $xx$ can be written in the form $z_2x'z_1\,z_2x'z_1$, where $z_1z_2=01010$. Then $x'$ is a square in $\mathbf{m}$, and thus $|x'|$ equals $3^k$ or $2\cdot3^k$ for some $k\ge0$ by Lemma~\ref{muprop},\,4. One can easily see that trying $|x'|=1,2,3,6$, it is impossible to obtain both squares $x'x'$ and $xx$. Hence, $x'$ must be a product of $\mu^2$-blocks ending with the block $\mu^2(0)$. Now we proceed with the case analysis.

\emph{Case 1}: both $uu$ and $vv$ contain $01010$ in the middle. Then
$$
w=z_2u'z_1\,z_2u'z_1\,z_4v'z_3\,z_4v'z_3\,z_2u'z_1\,z_2u'z_1, \text{ where } z_1z_2=z_3z_4=01010\,.
$$
Since $u'$ and $v'$ are products of $\mu^2$-blocks, we have $z_1z_4=z_3z_2=01010$. Hence, $u'u'v'v'u'u'$ is a factor of $\mathbf{m}$, a contradiction.

\emph{Case 2}: $uu$ contains $01010$, while $vv$ not. Then $w=z_2u'z_1z_2u'z_1\,vv\,z_2u'z_1z_2u'z_1$ and $u'$ is a product of $\mu^2$-blocks. Four subcases are possible depending on the existence of insertions on the borders of $u$ and $v$.

\emph{Case 2.1}: no insertions. Then $z_1vvz_2$ is a product of $\mu^2$-blocks, implying $|vv|\equiv4\pmod9$. By Lemma~\ref{muprop},\,4, $v=10$ and $z_1vvz_2=011\,010\,010$. But this is not a $\mu^2$-block, a contradiction.

\emph{Case 2.2}: an insertion only on the left. Then $v=z_2v'$. Let $\bar{v}=v'z_2$. Deleting all three insertions of $z_1z_2=01010$ from $w=z_2u'z_1z_2u'z_1\,z_2v'z_2v'\,z_2u'z_1z_2u'z_1$, one discovers the forbidden factor $u'u'\bar{v}\bar{v}u'u'$, contradicting the minimality of $w$. 

\emph{Case 2.3}: an insertion only on the right, is symmetric to Case 2.2. 

\emph{Case 2.4}: insertions on both sides. Then $v=z_2v'=v''z_1$, where $v'v''$ is a product of $\mu^2$-blocks. Hence $|vv|\equiv5\pmod9$, which is impossible by Lemma~\ref{muprop},\,4.

\emph{Case 3}: $vv$ contains $01010$, while $uu$ not. Note that in this case $u$ cannot have the suffix $0101$. Then 
\begin{itemize}
\item $w=uu\,z_2v'z_1z_2v'z_1\,uu$;
\item $v'$ is a product of $\mu^2$-blocks, ending with $\mu^2(0)$ (in particular, $v'=010\cdots1010$);
\item $uu$ is a factor of $\mathbf{m}$ (in particular, $u$ has no factor 0101).
\end{itemize}
If $|u|=1$, then either the first letter of $z_2$ or the last letter of $z_1$ equals $u$, implying that $w$ contains a cube of a letter, which is impossible. The assumption $|u|=2$ (i.\,e., $u=10$) also leads to a contradiction for all values of $z_1$. Namely, if $z_1$ ends with 0, then $uuz_2$ begins with $(10)^3$; if $z_1=01$, then $z_1uu=011010$ is not a valid beginning of a $\mu^2$-block; finally, $z_1=0101$ must be followed by 0, not by 1. Thus, $|u|\equiv0\pmod3$. Let us analyze the possible values of $z_1$.

\emph{Case 3.1}: $z_1=0101$, $v=0v'0101$, $w=uu0v'01010v'0101uu$. Since $u$ cannot end with 0101, $w$ has exactly two occurrences of $01010$ ($u[1]=0$). Let us put $v'=v''0$, $\bar{v}=0v''$. Deleting both occurrences of $01010$, we obtain the forbidden word $uu\bar{v}\bar{v}uu$ which is shorter than $w$, a contradiction.

\emph{Case 3.2}: $z_1=010$, $v=10v'010$, $w=uu10v'01010v'010uu$. If there no factor $01010$ on the left border of $vv$, then $uu$ ends in $\mathbf{m}$ in the position equal to 7 modulo 9. If this factor appears there, then $uu$ ends in $\mathbf{m}$ in the position equal to 3 modulo 9. Similarly, if there is the factor [resp., no factor] $01010$ on the right border of $vv$, then $uu$ begins in the position equal to 8 modulo 9 [resp., to 4 modulo 9]. Since $|u|=0\pmod3$, exactly one factor $01010$ should occur at the borders of $vv$. If this factor is on the left, we put $v'=010v''$. Then deleting both factors $01010$ we obtain a shorter forbidden factor $uuv''010v''010uu$ to get a contradiction (observe that the deleted suffix 010 of the second $u$ is replaced by the prefix $010$ of $v'$). Similarly, if the factor is on the right, we put $v'=v''10$ to obtain, after the deletion, a shorter forbidden factor $uu10v''10v''uu$.

For \emph{Case 3.3}: $z_1=01$ and \emph{Case 3.4}: $z_1=0$, the same analysis as in Case 3.2 works.

\emph{Case 4}: neither $uu$ nor $vv$ contains $01010$ in the middle. Then both $uu$ and $vv$ are factors of $\mathbf{m}$. We obtain contradictions between the length of $uu$ and its starting and ending positions in $\mathbf{m}$. 

\emph{Case 4.1}: $01010$ was inserted at the left border of $vv$. Since 01010 is followed by 010011, we see that either $v=010$ or $|v|\equiv0\pmod9$ by Lemma~\ref{muprop},\,4. In the first case, the starting position of $uu$ equals 4 modulo 9, and its ending position equals 2 modulo 9, contradicting Lemma~\ref{muprop},\,4. If $|v|\equiv0\pmod9$, let the starting position of $vv$ be equal to $k$ modulo 9. Then the ending position of $uu$ equals $k{+}4$ modulo 9, while its starting position equals either $k{-}5$ or $k$ modulo 9, depending on the existence of the factor 01010 at the right border of $vv$. In both cases, we have a contradiction with Lemma~\ref{muprop},\,4. 

\emph{Case 4.2}: $01010$ was inserted only at the right border of $vv$. Similar to Case~4.1, we analyze the possible lengths of $v$ ($|v|=1$, $|v|=6$, and $|v|\equiv0\pmod9$), obtaining that the length of $u$ cannot satisfy Lemma~\ref{muprop},\,4.

\smallskip
Thus, we finished the case study, obtaining contradictions in all cases. Hence, the forbidden word $w$ does not exist, and the set $S$ avoids $\{{\sf xxx, xxyyxx}\}$. Finally, we estimate the total number of factors in all words $\mathbf z\in S$, similar to the proof of Theorem~\ref{xyxyxx}. The word $\mathbf{m}$ has $\Theta(n)$ factors of length $n$; this follows, e.\,g., from Pansiot's classification theorem (see \cite{ChKa97}). It is clear that such a factor contains $n/2+O(1)$ zeroes and then, $n/18+O(1)$ factors $\mu^2(0)$. The latter quantity coincides with the number of places for insertions of the factor $01010$. Thus, from the factors of $\mathbf{m}$ of length $n$ we can construct $\Theta(n)2^{n/18+O(1)}$ factors of words from $S$. The lengths of these factors cover the interval of length $\Theta(n)$. Therefore, the growth rate of the binary language avoiding $\{{\sf xxx, xxyyxx}\}$ is at least $2^{1/18}$, as required.
\end{proof}

\subsection{Mapping ternary square-free words}

In this section, we explore another approach for getting lower bounds. Namely, the fact that the language of ternary square-free words has exponential growth leads to the following simple observation. 

\begin{bsrvtn}
If an $n$-uniform morphism $f:\{0,1,2\}^*\to\{0,1\}^*$ transforms any square-free ternary word to a binary word avoiding $\{{\sf xxx},P\}$, then the number of such binary words grows exponentially at rate at least $\alpha^{1/n}$, where $\alpha$ is the growth rate of the language of ternary square-free words.
\end{bsrvtn}

The morphisms with the desired properties can be obtained using the method described in~\cite{Och06}. The number $\alpha$ is known with a quite high precision: $1.3017597<\alpha<1.3017619$ (cf. \cite{Sh12rev}).

\begin{thrm} \label{other}
The number of binary cube-free words avoiding the pattern $P$, where $P\in\{\sf xxyxyy,xxyyxyx,xyxxyxy,xyxxyyxy \}$, grows exponentially with the rate of at least
\begin{itemize}
\item $\alpha^{1/14}\approx1.0190$ for $P=\sf xxyxyy$,
\item $\alpha^{1/13}\approx1.0205$ for $P=\sf xxyyxyx,xyxxyxy$,
\item $\alpha^{1/10}\approx1.0267$ for $P=\sf xyxxyyxy$.
\end{itemize}
\end{thrm}

\begin{proof}
In the proof of Theorem~\ref{avoid2} we used morphic preimages to reduce the proof of pattern avoidance to the exhaustive search of forbidden factors in short words. Since we cannot iterate morphisms acting on alphabets of different sizes, here we need a different argument for such a reduction. For this purpose, we construct binary words avoiding simultaneously cubes, the pattern $P$, and large squares. We use the notation $S_t$ for the $t$-ary pattern $({\sf x}_1\cdots{\sf x}_t)^2$. 

Consider the morphisms $g_1$, $g_2$, $g_3$, and $g_4$ such that
\begin{equation*}
\arraycolsep=4mm
\begin{array}{ll}
g_1(0)=01011001100101&g_2(0)=0100110011011\\
g_1(1)=00110110010011&g_2(1)=0100101101001\\
g_1(2)=00101001101011&g_2(2)=0011011001001\\[5pt]
g_3(0)=0010110110011&g_4(0)=0101100110\\
g_3(1)=0010110011011&g_4(1)=0101001011\\
g_3(2)=0010011010011&g_4(2)=0100110010.
\end{array}
\end{equation*}

For any square-free word $w\in\{0,1,2\}^*$ we claim that
\begin{itemize}
\item the word $g_1(w)$ avoids $\{{\sf xxx,xxyxyy},S_8\}$;
\item the word $g_2(w)$ avoids $\{{\sf xxx,xxyyxyx},S_9\}$;
\item the word $g_3(w)$ avoids $\{{\sf xxx,xyxxyxy},S_{10}\}$;
\item the word $g_4(w)$ avoids $\{{\sf xxx,xyxxyyxy},S_8\}$.
\end{itemize}

To prove this claim, we notice that for every binary pattern $P$ considered in this section, both variables $\sf x$ and $\sf y$ are involved in a square. This implies that in a word containing only squares of bounded length, potential occurrences of $P$ and of cubes have bounded length as well. So we can check exhaustively that $g_i(w)$ avoids cubes and $P$ for all short square-free words $w$. Let a \emph{large square} be an occurrence of $S_t$. There remains to prove that if $w$ is square-free, then $g_i(w)$ does not contain large squares. The proof is the same for all four morphisms.

Let $n_i=|g_i(a)|$, $a\in\{0,1,2\}$. First we check that the morphism $g_i$ is $2n_i$-syn\-chronizing. Indeed, any factor of $g_i(w)$ of length $2n_i$ contains a  $g_i$-image of some letter $a$; but it is easy to see that for any letters $a,b,c\in\{0,1,2\}$, the factor $g_i(a)$ only appears in $g_i(bc)$ as a prefix or as a suffix. Then we check that no large square appears in the $g_i$-image of a ternary square-free word of length 5. So, a potential large square $uu$ in $g_i(w)$ is such that $|u|>2n_i$ and thus $|u|=qn_i$ for some integer $q\ge 3$ by the synchronizing property. So $uu$ is contained in the image of a word of the form $w=avbvc$ with $a,b,c\in\Sigma_3$ and the center of $uu$ lies in $g_i(b)$. Moreover, $a\ne b$ and $b\ne c$ since $w$ is square-free. This implies that $abc$ is square-free and that $g_i(abc)$ contains a square $u'u'$ with $|u'|=n_i$. Now $u'u'$ is a large square because $n_i>t$ for all our morphisms $g_i$. This is a contradiction since no large square appears in the $g_i$-image of a ternary square-free word of length 5. The claim, and then the theorem, is proved.
\end{proof}

Proving Theorem~\ref{other}, we actually showed that the considered binary patterns can be avoided by binary cube-free words simultaneously with large squares. So, a natural problem is to find the exact bound for the length of these large squares. The following theorem gives this bound for all patterns listed in \eqref{mainlist}.

\begin{thrm} \label{largesq}
Let $P\in\{\sf xxyxyx,xxyxxy,xxyxyy,xxyyxx,xxyyxyx,xyxxyxy,xyxxyyxy\}$ and let $t(P)$ be the number such that the set of patterns $\{{\sf xxx},P,S_{t(P)}\}$ is 2-avoidable while the set $\{{\sf xxx},P,S_{t(P)-1}\}$ is 2-unavoidable. Then
$$
t(P)=\begin{cases}
4& \text{if } P=\sf xxyyxx,\\
5& \text{if } P\in\{\sf xxyxyy,xxyyxyx,xyxxyxy,xyxxyyxy\},\\
7& \text{if } P\in\{\sf xxyxxy, xxyxyx\}.
\end{cases}
$$
and the binary language avoiding $\{{\sf xxx},P,S_{t(P)}\}$ has exponential growth.
\end{thrm}

\begin{proof}
Below we list the morphisms mapping ternary square-free words to the binary words avoiding the required sets. The proof of avoidability and exponential growth is the same as for Theorem~\ref{other}.\\[3pt]
$\arraycolsep=1pt\begin{array}{l}
P={\sf xxyyxx},\text{ length}=62\\
0\to 00100101101100101001101101001001101011001010011011001001101011\\
1\to 00100101101100101001101100100110101100101001101101001001101011\\
2\to 00100101101100101001101011001001101100101001101101001001101011\\[3pt]
P={\sf xxyxyy},\text{ length}=88\\
0\to 00100110101100101001100110101100110010100110101100100110110010100\\
\phantom{0\to}11001101011001010011011\\
1\to 00100110101100101001100110101100100110110010100110101100110010100\\
\phantom{0\to}11001101011001010011011\\
2\to 00100110101100101001100110101100100110110010100110011010110010100\\
\phantom{0\to}11010110011001010011011\\[3pt]
P={\sf xyxxyxy},\text{ length}=49\\
0\to 0011001011011001101001001100110101100101001101011\\
1\to 0011001011011001101001001100101101100101001101011\\
2\to 0011001011011001001101011001010011011001001101011\\[3pt]
P={\sf xxyyxyx},\text{ length}=32\\
0\to 00100110110100100110011011010011\\
1\to 00100101101001001101101001011011\\
2\to 00100101100110110100100110011011\\[3pt]
P={\sf xyxxyyxy},\text{ length}=28\\
0\to 0010010110100110011010110011\\
1\to 0010010110100110010110110011\\
2\to 0010010110011011010010110011\\[3pt]
P={\sf xxyxyx},\text{ length}=44\\
0\to 00100110011010011001011001101001011011001101\\
1\to 00100110010110110011001011001101001100101101\\
2\to 00100110010110011010010110110011001011001101\\[3pt]
P={\sf xxyxxy},\text{ length}=66\\
0\to 001010011001011001101001100101101001101011001101001100101001101011\\
1\to 001010011001011001101001100101001101011001101001100101101001101011\\
2\to 001010011001011001101001011001010011010110011010011001011001101011
\end{array}$

Unavoidability of shorter squares is verified by computer search.
\end{proof}

\section{Growth rates: numerical results} \label{upp}

A general method to obtain upper bounds for the growth rates of factorial languages was proposed in \cite{Sh10tc}. An open-source implementation of this method can be found in \cite{Calc}. We adjust this method for each pattern under consideration and calculate the upper bounds for the growth rates of avoiding binary cube-free language. Here is a high-level overview of the method. 

Let $L$ be a factorial language and $M$ be its set of minimal forbidden words. If $L$ is an infinite language avoiding a pattern, then $M$ is also infinite. We construct a family $\{M_i\}$ of finite subsets of $M$ such that
$$
M_1\subseteq M_2\subseteq\dotsb\subseteq M_i\subseteq\dotsb\subseteq M,
\quad M_1\cup M_2\cup\dotsb\cup M_i\cup\dotsb=M. 
$$
Let $L_i$ be the binary factorial language with the set of minimal forbidden words $M_i$. One has
$$
L\subseteq\dotsb\subseteq L_i\subseteq\dotsb\subseteq L_1,
\quad L_1\cap L_2\cap\dotsb\cap L_i\cap\dotsb=L.
$$
It is not hard to show that the sequence of growth rates $\{\IR(L_i)\}$ decreases and converges to $\IR(L)$. The languages $L_i$ are regular, and then the number $\IR(L_i)$ can be found with any degree of precision. Increasing $i$, one can make the upper bound arbitrarily close to $\IR(L)$.

Thus, to obtain an upper bound for $\IR(L)$ one should make three steps. First, build a set of minimal forbidden words $M_i$ for the chosen $i$. Second,
convert this set into a deterministic finite automaton recognizing $L_i$ (the automaton should be both accessible and coaccessible). And finally, calculate the number $\IR(L_i)$. If we calculate $M_i$ by some search procedure and store it in a trie, then the second step can be implemented as a modified Aho-Corasick algorithm for pattern matching that converts the trie into an  automaton having the desired properties. At the third step we calculate the growth rate of $L_i$ with any prescribed precision by an efficient (linear in the size of automaton) iterative algorithm. The second and third steps are common for all factorial languages.

For each pattern we use an ad-hoc procedure for constructing the set of minimal forbidden words for avoiding languages. In most cases we bound the length of the constructed forbidden words with some constant. We iterate over the candidate forbidden words in the order of increasing length and check that they do not contain proper forbidden factors, using already built shorter forbidden words for pruning. In practice, the described method allows us to construct and handle sets of thousands of forbidden words and automata of millions of vertices efficiently. Some numerical results are presented in Table~\ref{tab1}. For each of the processed languages, the sequence of obtained upper bounds converges very fast. So, the actual value of the growth rate in each case is likely to be quite close to the given upper bound.

\begin{table}[!htb]
\caption{Growth rates of binary cube-free languages avoiding binary patterns: upper bounds} \label{tab1}
\centerline{ 
\begin{tabular}{|l|l|l|l|}
\hline
Pattern&Upper bound&Pattern&Upper bound\\
\hline
{\sf xxyxxy}&1.098891&{\sf xyxyx}&1 (previously known)\\
{\sf xxyxyy}&1.226850&{\sf xxyyxyx}&1.310975\\
{\sf xyxyxx}&1.138449&{\sf xyxxyxy}&1.281612\\
{\sf xxyyxx}&1.322304&{\sf xyxxyyxy}&1.348932\\
\hline
\end{tabular} }
\end{table}


\section{A language of polynomial growth}

Statement 3 of Theorem~\ref{main}, proved in Sect.~\ref{sect:low}, tells us that $\sf xyxyx$ is the only binary pattern that is avoided by a subexponentially-growing infinite set of binary cube-free words. In this section, we present two binary patterns $P_1$ and $P_2$ such that the binary language avoiding $\{{\sf xxx},P_1,P_2\}$ has polynomial growth. This language contains the binary overlap-free language and is incomparable with the binary $(7/3)$-free language (the latter one is the biggest binary $\beta$-free language of polynomial growth \cite{KaSh04}). Thus, this is an essentially new example of a language of polynomial growth.

\begin{thrm} \label{polyn}
The binary cube-free language avoiding both the patterns $\sf xyxyxx$ and $\sf xxyxyx$ has polynomial growth.
\end{thrm}

\begin{proof}
Let $L$ be the language of all binary cube-free words avoiding both $\sf xyxyxx$ and $\sf xxyxyx$. Obviously, both $L$ and its extendable part $\e(L)$ contain the set of all Thue-Morse factors. We aim to prove that this set coincides with $\e(L)$. The definition of extendable word implies that any word from $\e(L)$ is a factor of a Z-word all finite factors of which also belong to $\e(L)$.

For any word from $L$, the factors
$$
000,010101,010100,11001001,10010011,010010010,
$$
and their negations are forbidden. Hence, a word $w\in\e(L)$ has no factor $01010$, because any its extension to the right contains $010101$ or $010100$. Similarly, $w$ has no factor $00100$: extending this word, we inevitably meet one of the words $000$, $11001001$, $10010011$, or $(010)^3$. The same argument applies for 10101 and 11011.\\[4pt] 
\emph{Claim}. If a Z-word $\mathbf z$ has no factors $000$, $01010$, $00100$, and their negations, then $\mathbf z$ is a product of $\theta$-blocks.\\[4pt]
If two squares of letters in a word begin in positions of different parity, then this word surely contains one of the listed factors. To see this, just consider the closest pair of such squares. So, all squares of letters in $\mathbf z$ occur in positions of the same parity. Hence, one can factorize $\mathbf z$ into the factors of length 2 in a way that splits any square of a letter, thus getting the desired product.

\smallskip
Consider a Z-word $\mathbf z$ all factors of which belong to $\e(L)$. By the claim, $\mathbf z$ is a product of 1-blocks. Consider its Thue-Morse preimage $\mathbf{z}'=\theta^{-1}(\mathbf z)$. The Z-word $\mathbf{z}'$ avoids the patterns $\sf xxx,xxyxyx$, and $\sf xyxyxx$. Indeed, if $\mathbf{z}'$ contains an image of a pattern under $f$, then $\mathbf{z}$ contains an image of the same pattern under $\theta\!f$. Hence, $\mathbf{z}'$ has no factors listed in the claim, and we conclude that it is a product of $\theta$-blocks. Then $\mathbf z$ is a product of $\theta^2$-blocks. Repeating this argument inductively, we obtain that $\mathbf z$ is a product of $\theta^n$-blocks for any $n$. Therefore, any finite factor of  $\mathbf z$ is a factor of some $\theta^n$-block, i.e., a Thue-Morse factor, as desired.

\smallskip
The set of Thue-Morse factors contains $\Theta(n)$ words of length $n$, and then has the growth rate 1. But the languages $L$ and $\e(L)$ always have the same growth rate (see \cite[Theorem~3.1]{Sh08ita}), so our language $L$ grows subexponentially. To prove that this growth is polynomial, some additional work is needed. 

Let us take an overlap $w=0v0v0\in L$ with $|v|>2$ and analyze how it can be extended within $L$. The words $0w,w0$ are images of $\sf xxyxyx$ and $\sf xyxyxx$, respectively, so, $0w,w0\notin L$. Note that $v$ begins or ends with 1, because $w$ has no factor 000. Assuming w.l.o.g. that $v=1v'$ and extending $w$ to the right by one symbol, we get a longer overlap: $w1=01v'01v'01$. We see that $w11$ and $w101$ are images of $\sf xyxyxx$. Assume that $w100=01v'01v'0100\in L$. 

If the last letter of $v'$ is 1, then $v$ ends with 11, because $010100\notin L$. Then $v'$ cannot begin with 1, because the factor $11011$ in the middle of $w$ means that $w$ contains a forbidden factor (compare to the beginning of the proof). But if $v'=0v''$, we see that the word $w100=01\,0v''010v''0100$ meets the pattern $\sf xyxyxx$ ($x\to 0$, $y\to v''01$). So, $v$ ends with 0 and then with 10. Then the word $w1001$ ends with 1001001, guaranteeing that  $w10011,w10010\notin L$. Thus, we have proved the following property.

\begin{itemize}
\item[($\blacktriangle$)] Suppose that $w=uvuvu\in L$, $|uv|\ge4$, and $|u|>1$. Then $w$ can be extended within $L$ by at most three letters to each side.
\end{itemize}

Finally, we estimate the number of words in $L$ that are not $(7/3)$-free. These words contain overlaps with $|u|\ge|v|/2$. From ($\blacktriangle$) it follows that the set of words in $L$ containing overlaps such that $|u|>1$ and $|u|\ge|v|/2$, is finite. So, it remains to consider the case $|u|=1$ (and then $|v|\le2$). If $|uv|=2$, the overlap is 01010 or 10101. It cannot be extended within $L$. Now let $|uv|=3$. Such an overlap must contain the factor 00100 or 11011, which cannot be extended within $L$ to both sides simultaneously by more than one letter. Then the words from $L$ containing an overlap of period 3 have the form~$0010010z$ or $10010010z$ up to reversal and negation. The number of such words grows polynomially because the word $z$ is overlap-free. Since the number of $(7/3)$-free words is also polynomial, we get a polynomial upper bound on the number of words in $L$.
\end{proof}

\begin{rmrk}
Concerning the bounds for the degree of the polynomial growth of the language $L$ considered in Theorem~\ref{polyn}, we have shown, in fact, that one can take the upper bound derived for the $(7/3)$-free language in \cite{BCJ09}. The obvious lower bound stems from the fact that $L$ contains all overlap-free words. 
\end{rmrk}

\bibliographystyle{plain}
\bibliography{my_bib}

\end{document}